\documentclass[a4paper,12pt]{article}%
\usepackage{amsmath,amsfonts}
\usepackage{theorem}
\usepackage{indentfirst}
\usepackage{graphicx}
\usepackage{longtable}
\usepackage{amssymb}
\usepackage{amsmath}
\usepackage{color}
\usepackage{amsfonts}%
\providecommand{\U}[1]{\protect\rule{.1in}{.1in}}

\theoremstyle{plain} { \theorembodyfont{\rmfamily}

}
\newtheorem{proposition}{Proposition}

\newtheorem{theorem}{Theorem}
\newcommand{\qed}{\hfill \mbox{\raggedright \rule{.07in}{.1in}}}
\newenvironment{proof}{\vspace{1ex}\noindent{\bf Proof}\hspace{0.5em}}
{\hfill\qed\vspace{1ex}}

\begin{document}

\title{Distributionally Robust Mean-Variance Portfolio Selection with Wasserstein
Distances\thanks{X. Y. Zhou gratefully acknowledges financial support through
a start-up grant at Columbia University and through the FDT Center for
Intelligent Asset Management. Blanchet gratefully acknowledges support from the National Science Foundation.}}
\author{Jose Blanchet\thanks{Department of Management Science and Engineering,
Stanford University, Stanford, California 94305, USA,
\texttt{jblanche@stanford.edu}}\ \ \ Lin Chen\thanks{Department of Industrial Engineering and Operations Research,
Columbia University, New York, New York 10027, USA,
\texttt{lc3110@columbia.edu}}\ \ \ Xun Yu
Zhou\thanks{Department of Industrial Engineering and Operations Research,
Columbia University, New York, New York 10027, USA,
\texttt{xz2574@columbia.edu}}}
\maketitle

\begin{abstract}
We revisit Markowitz's mean-variance portfolio selection model by considering
a distributionally robust version, where the region of distributional
uncertainty is around the empirical measure and the discrepancy between probability
measures is dictated by the so-called Wasserstein distance. We reduce this
problem into an empirical variance minimization problem with an additional
regularization term.
Moreover, we extend recent inference methodology in order to select the
size of the distributional uncertainty as well as the associated robust target return rate in a
data-driven way.

\bigskip


\end{abstract}

\section{Introduction}

We study data-driven mean--variance portfolio selection with model uncertainty
(or ambiguity). The classical Markowitz mean--variance model (Markowitz 1952)
is to choose a portfolio weighting vector $\phi\in\mathbb{R}^{d}$ (all the
vectors in this paper are, by convention, columns) among $d$ stocks to
maximize the risk-adjusted expected return. The precise formulation
is\footnote{There are several mathematically equivalent formulations of the
original mean--variance model.}
\begin{equation}
\min_{\phi\in\mathbb{R}^{d}}\{\phi^{T}Var_{P^{\ast}}\left(  R\right)
\phi:\phi^{T}1=1,\;\phi^{T}E_{P^{\ast}}\left(  R\right)  =\rho\}, \label{Op}%
\end{equation}
where $R$ is the $d$-dimensional vector of random returns of the stocks;
$P^{\ast}$ is the probability measure underlying the distribution of $R$;
$E_{P^{\ast}}$ and $Var_{P^{\ast}}$ are respectively the expectation and
variance under $P^{\ast}$; and $\rho$ is the targeted expected return of the portfolio.

%

It is well known that this model has a major drawback. On one hand,
its solutions are very sensitive to the underlying parameters, namely the mean
and the covariance matrix of the stocks. On the other hand, $E_{P^{\ast}}$ is
unknown in practice; so one has to resort to the empirical versions of the
mean and the covariance matrix instead, which are usually significantly
deviated from the true ones (especially the mean, due to the notorious
\textquotedblleft mean blur\textquotedblright\ problem).

This motivates the development of the \textquotedblleft
robust\textquotedblright\ formulation of the Markowitz model which recognizes
and tries to account for the impact of the (potentially significant)
discrepancies between $P^{\ast}$ and its empirical version. This idea
originates in the robust control approach in control theory (see Peterson, James and Dupuis (2000)). Hansen and Sargent (2008) give a systematic
account on applications of robust control to economic models. There is also a
rich literature on robustification of portfolio choice. Lobo and Boyd (2000)
are among the first to provide a worst-case robust analysis with respect to the
second-order moment uncertainty within the Markowitz framework. Pflug and
Wozabal (2007) formulate a Markowitz model with distributional robustness
based on a Wasserstein distance, a metric measuring the discrepancy between
two probability measures which we also apply in this paper. Nevertheless,
their formulation involves an additional value-at-risk type of constraint that
leads to a much more complex optimization problem. More importantly, their
choice of the uncertainty size is exogenous and no guidance for optimally
selecting the size is given. Esfahani and Kuhn (2017) provide representations
for the worst-case expectations in a Wasserstein-based ambiguity set centered
at the empirical measure, and then apply their results to portfolio selection
using different risk measures, leading to models different from the Markowitz
model. The choice of the uncertainty size is suboptimal because it dramatically deteriorates
with the dimension of the underlying portfolio.

Along a line different from full distributional uncertainty, Delage and Ye
(2010) construct uncertainty regions only involving means and covariances of
the return vector. Wozabal (2012) also considers a robust portfolio model with
risk constraints based on expected short-fall, resulting in an optimization
problem that requires solving multiple convex problems.
Again, these papers do not consider the choice of the size of the uncertainty sets.

Papers that address different optimization techniques (such as interior point
methods, conic programming and linear matrix inequalities) in solving robust
portfolio selection include Halldorsson and Tutuncu (2000), Costa and Paiva
(2001) and Ghaoui, Oks and Oustry (2003). Finally, we mention the works of Goh
and Sim (2010) and Wisesmann, Kuhn and Sim (2014) who investigate different
forms of distributional ambiguity sets, as well as those of Goh and Sim
(2013) and Jiang and Guan (2016) who study distributional robust
formulations based on the Kullback-Leibler divergences. It is worth noting
that the Kullback-Leibler divergence-based formulation is popular in
economics (see Hansen and Sargent (2008)).

In this paper, we are interested in studying a distributionally robust
optimization (DRO) formulation of the mean--variance problem, given by
\begin{equation}
\min_{\phi\in\mathcal{F}_{\delta,\bar{\alpha}}\left(  n\right) }\;\;\max
_{P\in\mathcal{U}_{\delta}(P_{n})}\{\phi^{T}Var_{P}\left(  R\right)  \phi\},
\label{Pp}%
\end{equation}
where $P_{n}$ is the empirical probability derived from historical information
of the sample size $n$, $\mathcal{U}_{\delta}(P_{n}):=\{P:D_{c}(P,P_{n}%
)\leq\delta\}$ is the ambiguity set, $\mathcal{F}_{\delta,\bar{\alpha}}\left(
n\right) =\{\phi:\phi^{T}1=1,\min_{P\in\mathcal{U}_{\delta}(P_{n})}[E_{P}\left(  \phi^{T}R\right)
 ]\geq\bar{\alpha}\}$ is the feasible region
of portfolios, $E_{P}$ and $Var_{P}\left(  R\right)  $ denote respectively the
mean and the covariance matrix under $P$,
and $D_{c}\left(  \cdot,\cdot\right)  $
is a notion of discrepancy between two probability measures based on a
suitably defined Wasserstein distance.\footnote{Recent work by Blanchet, Kang
and Murthy (2016) shows that a similar definition of discrepancy in some
other models recovers exactly some well-known machine learning algorithms, such
as square-root Lasso and support vector machines.}

Intuitively, formulation (\ref{Pp}) introduces an artificial adversary $P$
(whose problem is that of the inner maximization) as a tool to account for the
impact of the model uncertainty around the empirical distribution. There are
two key parameters, $\delta$ and $\bar{\alpha}$, in this formulation, and they
need to be carefully chosen. The parameter $\delta$ can be interpreted as the
power given to the adversary: The larger the value of $\delta$ the more power
is given. If $\delta$ is too large relative to the evidence (i.e. the size of
$n$), then the portfolio selection will tend to be unnecessarily conservative.
On the other hand, $\bar{\alpha}$ can be regraded as the lowest acceptable
target return given the ambiguity set. Naturally, the choice of $\bar{\alpha}$
should be based on the original target $\rho$ given in (1); but one also needs
to take into account the size of the distributional uncertainty, $\delta$.
Using $\bar{\alpha}=\rho$ will tend to generate portfolios that are too
aggressive; it is more sensible to choose $\bar{\alpha}<\rho$ in a way such
that $\rho-\bar{\alpha}$ is naturally informed by $\delta$.


This paper makes two main contributions. First, we show that
(\ref{Pp}) is equivalent to an (explicitly formulated) non-robust minimization
problem in terms of the empirical probability measure in which a proper
penalty term or \textquotedblleft regularization term\textquotedblright\ is
added to the objective function. This connects to the {\it direct} use of
regularization in variance minimization techniques widely employed both in the
machine learning literature and in practice.
Indeed, practitioners who use mean--variance portfolio selection models often
introduce regularization penalties, inspired by Lasso, in order to enhance the
sparsity leading to fewer stocks in their portfolios. Our use of Wasserstein
distance to model distributional uncertainty naturally gives rise to a
regularization term, suggesting an alternative, yet theoretical, justification
for its use in practice.
Our result shows that our robust strategies are able to enhance out-of-sample
performance with basically the same level of computational tractability as
standard mean-variance selection.

Our second main contribution provides guidance on the choice of the size of
the ambiguity set, $\delta$, as well as that of the worst mean return target, $\bar\alpha$.
 This is accomplished by adapting and extending
the robust Wasserstein profile inference (RWPI) framework, recently introduced
and developed by Blanchet, Kang and Murthy (2016), in a data-driven way
that combines optimization principles and basic statistical theory under suitable
mixing conditions on historical data.

The rest of the paper is organized as follows: We split Section \ref{Sec_Main}
into three parts, with the assumptions discussed in the first part,
followed by the tractability of our distributionally robust optimization
formulation (culminating with Theorem \ref{ThmMain1}), and the choice of
distributional uncertainty (see Theorem \ref{ThmMain2} and Section
\ref{Sec_Choice_Alpha}).
Some conclusions and extensions are
given in Section \ref{Sec_Conc}. The technical proofs of our results are given
in various appendices at the end of the paper.

\section{Formulation and Main Results\label{Sec_Main}}

\subsection{Basic notation and assumptions}

In this subsection, we introduce our assumptions and notation, and review some
useful concepts.

The mathematical formulation of our problem is given by (\ref{Pp}); but we now
need to specify the critical measure $D_{c}\left(  \cdot\right)  $. Let
$\mathcal{P}(\mathbb{R}^{d}\times\mathbb{R}^{d})$ be the space of Borel
probability measures supported on $\mathbb{R}^{d}\times\mathbb{R}^{d}$. A
given element $\pi\in\mathcal{P}(\mathbb{R}^{d}\times\mathbb{R}^{d})$ is
associated to a random vector $\left(  U,V\right)  $, where $U\in
\mathbb{R}^{d}$ and $V\in\mathbb{R}^{d}$, in the following way: $\pi
_{U}\left(  A\right)  =\pi\left(  A\times\mathbb{R}^{d}\right)  $ and $\pi
_{V}\left(  A\right)  =\pi\left(  \mathbb{R}^{d}\times A\right)  $ for every
Borel set $A\subset\mathbb{R}^{d}$, where $\pi_{U}$ and $\pi_{V}$ are
respectively the distributions of $U$ and $\pi_{V}$.

Let us introduce a cost function $c:\mathbb{R}^{d}\times\mathbb{R}%
^{d}\rightarrow\lbrack0,\infty]$, which we shall assume to be lower
semicontinuous and such that $c\left(  u,u\right)  =0$ for any $u\in
\mathbb{R}^{d}$.

Now, given two probability distribution $P$ and $Q$ supported on
$\mathbb{R}^{d}$ and a cost function $c$, define
\begin{equation}
D_{c}(P,Q):=\inf\{E_{\pi}[c(U,W)]:\pi\in\mathcal{P}(\mathbb{R}^{d}%
\times\mathbb{R}^{d}),\pi_{U}=P,\pi_{W}=Q\},
\end{equation}
which can be interpreted as the optimal (minimal) transportation cost (also
known as the optimal transport discrepancy or the Wasserstein discrepancy) of
moving the mass from $P$ into the mass of $Q$ under a cost $c\left(
x,y\right)  $ per unit of mass transported from $x$ to $y$. If for a given
$p>0$, $c^{1/p}\left(  \cdot\right)  $ is a metric, then so is $D_{c}^{1/p}$
(see Villani (2003)). Such a metric $D_{c}^{1/p}$ is known as a Wasserstein
distance of order $p$. Most of the times in this paper, we choose the following cost
function
\begin{equation}
c(u,w)=||w-u||_{q}^{2}%
\end{equation}
where $q\geq1$ is fixed (which leads to a Wasserstein distance of order $2$%
).\footnote{Different cost functions can be used, resulting in different
regularization penalties, as we will discuss in Section 5: Concluding Remarks.}
Finally, we define the ambiguity set $\mathcal{U}_{\delta}(P_{n})$ as
\[
\mathcal{U}_{\delta}(P_{n})=\{P:D_{c}(P,P_{n})\leq\delta\},
\]
where $P_{n}$ is the empirical probability measure with a sample size $n$,
i.e
\[
P_{n}(dr)=\frac{1}{n}\sum\limits_{i=1}^{n}\delta_{R_{i}}(dr)
\]
where $R_{i}$ $(i=1,2,...,n)$ are realizations of $R$ and $\delta_{R_{i}}%
(\cdot)$ is the indicator function.

\subsection{Computational tractability}

We now reformulate (\ref{Pp}) in a way that becomes computationally tractable.
The first step is to show that the feasible region over $\phi$ in the outer
minimization part can be explicitly evaluated. This is given in the following
proposition, whose proof is relegated to the Appendix.

\begin{proposition}
\label{prop1} For $c(u,w)=||u-w||_{q}^{2}$, $q\geq1$, we have
\begin{equation}
\min\limits_{P\in\mathcal{U}_{\delta}(P_{n})}E_{P}(\phi^{T}R)=E_{P_{n}}%
(\phi^{T}R)-\sqrt{\delta}||\phi||_{p}, \label{outermin}%
\end{equation}
with $1/p+1/q=1$.
\end{proposition}

Therefore, the feasible region is equivalent to
\[
\mathcal{F}_{\delta,\bar{\alpha}}\left(  n\right)  =\{\phi:\phi^{T}%
1=1,E_{P_{n}}(\phi^{T}R)\geq\bar{\alpha}+\sqrt{\delta}||\phi||_{p}\},
\]
which can now be seen as clearly convex.

Next, by fixing $E_{P}(\phi^{T} R)=\alpha\geq\bar\alpha$ in the inner
maximization portion of problem (\ref{Pp}) we obtain the following equivalent
reformulation%
\begin{equation}
\min_{\phi\in\mathcal{F}_{\delta,\bar{\alpha}}}\left\{  \max_{\alpha\geq
\bar\alpha}\left[  \max_{P\in\mathcal{U}_{\delta}(P_{n}),E_{P}(\phi^{T} R)
=\alpha}\{\phi^{T}E_{P}\left(  RR^{T}\right)  \phi\}-\alpha^{2}\right]
\right\}  . \label{Ra}%
\end{equation}
Introducing $E_{P}(\phi^{T} R)=\alpha$ is useful because the inner-most
maximization problem is now linear in $P$. So, let us concentrate on the
problem
\begin{equation}
\max\limits_{P\in\mathcal{U}_{\delta}(P_{n}),E_{P}(\phi^{T} R)=\alpha}\phi
^{T}E_{P}[RR^{T}]\phi. \label{maxpp}%
\end{equation}
The following proposition solves this problem in terms of a general cost
function $c$.

\begin{proposition}
\label{p1}For an arbitrary cost function $c$ that is lower semicontinuous and
non-negative, the optimal value function of problem (\ref{maxpp}) is given by
\begin{equation}
\label{dualpp}\inf\limits_{\lambda_{1}\geq0,\lambda_{2}}\left[ \frac{1}{n}%
\sum\limits_{i=1}^{n}\Phi(R_{i})+\lambda_{1}\delta+\lambda_{2}\alpha\right] ,
\end{equation}
where
\[
\Phi(R_{i}):=\sup\limits_{u}[(\phi^{T}u)^{2}-\lambda_{1}c(u,R_{i})-\lambda
_{2}\phi^{T}u].
\]

\end{proposition}

A proof, based on a dual argument, is given in the Appendix. Thanks to this
proposition, we are able to reduce the inner (infinite dimensional) variational
problem in (\ref{Pp}) into a two-dimensional optimization problem in terms of
$\lambda_{1}$ and $\lambda_{2}$, which can be further simplified if the cost
function $c$ has additional structure. We make this statement precise in the
case of a quadratic $l_{q}$ cost.

\begin{proposition}
\label{Prop_haq} Let $c(u,w)=||u-w||_{q}^{2}$ with $q\geq1$ and $1/p+1/q=1$.
If $(\alpha-\phi^{T}E_{P_{n}}[R])^{2}-\delta||\phi||_{p}^{2}\leq0$, then the
value of (\ref{maxpp}) is equal to
\begin{align*}
h(\alpha,\phi):=  &  E_{P_{n}}\left[  (\phi^{T}R)^{2}\right]  +2(\alpha
-\phi^{T}E_{P_{n}}[R])\phi^{T}E_{P_{n}}[R]+\delta||\phi||_{p}^{2}\\
&  +2\sqrt{\delta||\phi||_{p}^{2}-(\alpha-\phi^{T}E_{P_{n}}[R])^{2}}\sqrt
{\phi^{T}Var_{P_{n}}\left(  R\right)  \phi}.
\end{align*}

\end{proposition}

Again, a proof of this proposition is in the Appendix. The condition
$(\alpha-\phi^{T}E_{P_{n}}[R])^{2}-\delta||\phi||_{p}^{2}\leq0$ is to make
sure that (\ref{maxpp}) is feasible, failing which the optimal value
$h(\alpha,\phi)=-\infty$. Proposition \ref{Prop_haq} ultimately leads to the
following main result of the paper, one that transforms (\ref{Pp}) into a
non-robust portfolio selection problem in terms of the empirical measure
$P_{n}$.

\begin{theorem}
\label{ThmMain1}The primal formulation given in (\ref{Pp}) is equivalent to
the following dual problem
\begin{equation}
\min_{\phi\in\mathcal{F}_{\delta,\bar{\alpha}}\left(  n\right)  } \left(
\sqrt{\phi^{T}Var_{P_{n}}\left(  R\right)  \phi}+\sqrt{\delta}||\phi
||_{p}\right)  ^{2}, \label{Ppp}%
\end{equation}
in the sense that the two problems have the same optimal solutions and optimal
value.
\end{theorem}

\begin{proof}
Note that
\begin{align*}
&  h(\alpha,\phi)-\alpha^{2}\\
&  =E_{P_{n}}\left[  (\phi^{T}R)^{2}\right]  +2(\alpha-\phi^{T}E_{P_{n}%
}[R])\phi^{T}E_{P_{n}}[R]-\alpha^{2}+\delta||\phi||_{p}^{2}\\
&  +2\sqrt{\delta||\phi||_{p}^{2}-(\alpha-\phi^{T}E_{P_{n}}[R])^{2}}\sqrt
{\phi^{T}Var_{P_{n}}\left(  R\right)  \phi}\\
&  =E_{P_{n}}\left[  (\phi^{T}R)^{2}\right]  +2\alpha\phi^{T}E_{P_{n}%
}[R]-(\phi^{T}E_{P_{n}}[R])^{2}-\alpha^{2}-(\phi^{T}E_{P_{n}}[R])^{2}%
+\delta||\phi||_{p}^{2}\\
&  +2\sqrt{\delta||\phi||_{p}^{2}-(\alpha-\phi^{T}E_{P_{n}}[R])^{2}}\sqrt
{\phi^{T}Var_{P_{n}}\left(  R\right)  \phi}\\
&  =\phi^{T}Var_{P_{n}}\left(  R\right)  \phi+\{\delta||\phi||_{p}^{2}%
-(\alpha-\phi^{T}E_{P_{n}}[R])^{2}\}\\
&  +2\sqrt{\delta||\phi||_{p}^{2}-(\alpha-\phi^{T}E_{P_{n}}[R])^{2}}\sqrt
{\phi^{T}Var_{P_{n}}\left(  R\right)  \phi}\\
&  =\left(  \sqrt{\phi^{T}Var_{P_{n}}\left(  R\right)  \phi}+\sqrt
{\delta||\phi||_{p}^{2}-(\alpha-\phi^{T}E_{P_{n}}[R])^{2}}\right)  ^{2}.
\end{align*}
Therefore, it follows from Proposition \ref{Prop_haq} that
\[
\max_{\alpha\geq\bar\alpha, (\alpha-\phi^{T}E_{P_{n}}[R])^{2}-\delta
||\phi||_{p}^{2}\leq0}\left[ h(\alpha,\phi) -\alpha^{2}\right]  =\left(
\sqrt{\phi^{T}Var_{P_{n}}\left(  R\right)  \phi}+\sqrt{\delta}||\phi
||_{p}\right)  ^{2},
\]
with the optimal $\alpha_{opt}=\phi^{T}E_{P_{n}}[R]\geq\bar\alpha$. This
concludes the result.
\end{proof}

Because the mapping $\phi\rightarrow\phi^{T}Var_{P}\left(  R\right)  \phi$ is
convex and the feasible region $\mathcal{F}_{\delta,\bar{\alpha}}\left(
n\right)  $ is convex, (\ref{Ppp}) and therefore (\ref{Pp}) are both convex
optimization problems. As such, they are tractable optimization problems.
Moreover, our approach justifies the regularization technique in
mean--variance portfolio selection often adopted in practical settings.

\section{Choice of Model Parameters}

There are two key parameters, $\delta$ and $\bar{\alpha}$, in the formulation
(\ref{Pp}), the choice of which is not only curious in theory, but also
crucial in practical implementation and for the success of our algorithm. The
idea is that the choice of these parameters should be informed by the data
(i.e. in a data-driven way) based on some statistical principles, rather than
arbitrarily exogenous. Specifically, we define the distributional uncertainty
region just large enough so that the correct optimal portfolio (the one which
we would apply if the underlying distribution was known) becomes a plausible
choice with a sufficiently high confidence level.

We need to impose several statistical assumptions.

\bigskip

\textbf{A1)} The underlying return time series $(R_{k}:k\geq0)$
is a stationary, ergodic process satisfying $E_{P^{\ast}}\left(  ||R_{k}%
||_{2}^{4}\right)  <\infty$ for each $k\geq0$. Moreover, for each measurable
$g\left(  \cdot\right)  $ such that $\left\vert g\left(  x\right)  \right\vert
\leq c(1+\left\Vert x\right\Vert _{2}^{2})$ for some $c>0$, the limit
\[
\Upsilon_{g}:=\lim_{n\rightarrow\infty}Var_{P^{\ast}}\left(  n^{-1/2}%
\sum_{k=1}^{n}g\left(  R_{k}\right)  \right)
\]
exists and
\[
n^{1/2}\left[  E_{P_{n}}\left(  g\left(  R\right)  \right)  -E_{P^{\ast}%
}\left(  g\left(  R\right)  \right)  \right]  \Rightarrow N\left(
0,\Upsilon_{g}\right)  ,
\]
where (and from that time on) ``$\Rightarrow$" denotes weak convergence.\newline

\bigskip

\textbf{A2)} For any matrix $\Lambda\in\mathbb{R}^{d\times d}$ and any vector
$\zeta\in\mathbb{R}^{d}$ such that either $\Lambda\neq0$ or $\zeta\neq0$,
\[
P^{\ast}\left(  \left\Vert \Lambda R+\zeta\right\Vert _{2}>0\right)  >0.
\]
\bigskip

\bigskip

\textbf{A3) } The classical model (\ref{Op}) has a unique solution $\phi
^{\ast}$. Moreover, $Var_{P^{\ast}}\left[  E_{P^{\ast}}\left(  R^{T}\right)
R\right]  >0.$

\bigskip

Assumption A1)
is standard for most time series models (after removing seasonality).
Assumption A2) can be easily checked assuming, for example, that $R$ has a
density. Assumption A3) is a technical assumption which can be relaxed, but
then the evaluation of the optimal choice of $\delta$ would become more
cumbersome, as we shall explain.

\subsection{Choice of $\delta$}

The choice of the uncertainty size $\delta$ is crucial. If $\delta$ is too large, then
there is too much model ambiguity and the available data becomes less relevant.
In this case the resulting optimal portfolios will tend to be just equal allocations.
If $\delta$ is too small, then the effect of
robustification will be negligible. Therefore, the choice of $\delta$ should {\it not} be exogenously
specified; rather it should be endogenously informed by the data.

Theorem \ref{ThmMain1} actually suggests
an appropriate order of
$\delta=\delta_{n}$ (here $n$ is the size of the
available return time series data) in terms of $n$. Because the differences between the
optimal value obtained by solving (\ref{Op}) and that obtained by solving
the empirical version of (\ref{Op}) are of order $O\left(  n^{-1/2}\right) $, it
follows from Theorem \ref{ThmMain1} that any choice of $\delta_{n}$ in the
order of $o\left(  n^{-1}\right) $ would be too small. Hence, an ``optimal" order of
$\delta_{n}$ should be $O\left(  n^{-1}\right) $.

In order to choose an appropriate $\delta_{n}$, here we follow the idea behind
the RWPI approach introduced in
Blanchet, Kang and Murthy (2016).

Intuitively, the set $\mathcal{U}_{\delta}(P_{n})=\{P:D_{c}(P,P_{n})\leq
\delta\}$ contains all the probability measures that are plausible variations
of the data represented by $P_{n}$. Denote by $\mathcal{Q}\left(  P\right)  $
the Markowitz portfolio selection problem with target return $\rho$ assuming
that $P$ is the underlying model:
\begin{align}
&  \min\limits_{1^{T}\phi=1}\ \ \phi^{T}E_{P}[RR^{T}]\phi\label{Q_p_problem}\\
&  \left.  s.t\ \ \phi^{T}E_{P}[R]=\rho,\right. \nonumber
\end{align}
and by $\phi_{P}$ a solution to $\mathcal{Q}\left(  P\right)  $ and $\Phi_{P}$
the set of all such solutions. According to Assumption A3) we have
$\Phi_{P^{\ast}}=\left\{  \phi^{\ast}\right\} $ for some portfolio $\phi
^{\ast}$. Therefore there exist (unique) Lagrange multipliers $\lambda_{1}^{\ast}$
and $\lambda_{2}^{\ast}$ such that
\begin{align}
2E_{P^{\ast}}(RR^{T})\phi^{\ast}-\lambda_{1}^{\ast}E_{P^{\ast}}[R]-\lambda
_{2}^{\ast}1  &  =0,\label{LagStar}\\
(\phi^{\ast})^{T}E_{P^{*}}[R]-\rho &  =0.\nonumber
\end{align}

Because $P_{n}\Rightarrow P_{\ast}$ under mild assumptions (e.g. A1)), it is
reasonable to estimate $\phi^{\ast}$ by computing $\phi_{P_{n}}$. Now, when
$\delta$ is suitably chosen, since $\mathcal{U}_{\delta}(P_{n})$ constitutes
the models that are plausible variations of $P_{n}$, any $\phi_{P}$ with
$P\in\mathcal{U}_{\delta}(P_{n})$ is a plausible estimate of $\phi^{\ast}$.
This intuition motivates the definition of the following set
\[
\Lambda_{\delta}(P_{n})=\cup_{P\in\mathcal{U}_{\delta}(P_{n})}\Phi_{P},
\]
which corresponds to all the plausible estimates of $\phi^{\ast}$. As a
result, $\Lambda_{\delta}(P_{n})$ is a natural confidence region for
$\phi^{\ast}$ and, therefore, $\delta$ should be chosen as the smallest number
$\delta_{n}^{\ast}$ such that $\phi^{\ast}$ belongs to this region with a
given confidence level. Namely,
\[
\delta_{n}^{\ast}=\min\{\delta:P^{*}\left(  \phi^{\ast}\in\Lambda_{\delta
}(P_{n})\right)  \geq1-\delta_{0}\},
\]
where $1-\delta_{0}$ is a user-defined confidence level (typically 95\%).

However, by the mere definition, it is difficult to compute $\delta_{n}^{\ast}$. We now
provide a simpler representation for $\delta_{n}^{\ast}$ via an auxiliary
function called the robust Wasserstein profile (RWP) function. To this end,
first observe that any $\phi\in\Lambda_{\delta}\left(  P_{n}\right)  $ if and
only if there exist $P\in\mathcal{U}_{\delta}(P_{n})$ and $\lambda_{1}%
,\lambda_{2}\in\left(  -\infty,\infty\right)  $ such that%
\begin{align*}
2E_{P}(RR^{T})\phi-\lambda_{1}E_{P}[R]-\lambda_{2}1  &  =0,\\
\phi^T E_{P}(R)-\rho  &  =0.
\end{align*}
From these two equations, multiplying the first equation by $\phi$,
substituting the expression in the second equation and noting that $\phi
\cdot1=1$, we obtain
\[
\lambda_{2}=2\left(  \phi\right)  ^{T}E_{P}(RR^{T})\phi-\lambda_{1}\rho.
\]

We now define the following RWP function
\[
\mathcal{\bar{R}}_{n}(\phi,\lambda_{1},\Sigma,\mu):=\inf\left\{ D_{c}%
(P,P_{n}):%
\begin{cases}
2\Sigma\phi-\lambda_{1}\mu=\left(  2\left(  \phi\right)  ^{T}\Sigma
\phi-\lambda_{1}\mu\cdot\phi\right)  1\\
\mu=E_{P}[R],\Sigma=E_{P}(RR^{T})
\end{cases}
\right\} ,
\]
for $(\phi,\lambda_{1},\Sigma,\mu)\in\mathbb{R}^{d}\times\mathbb{R}%
\times\mathcal{S}^{d\times d}_{+}\times\mathbb{R}^{d}$ where $\mathcal{S}%
^{d\times d}_{+}$ is the set of all the symmetric positive semidefinite
matrices and we convent that $\inf\emptyset:=+\infty$. Moreover, we define
\[
\mathcal{\bar{R}}_{n}^{\ast}(\phi^{\ast}):=\min_{\Sigma\in\mathcal{S}^{d\times
d}_{+} ,\mu\in\mathbb{R}^{d},\lambda_{1}\in\mathbb{R}}\mathcal{\bar{R}}%
_{n}(\phi^{\ast},\lambda_{1},\Sigma,\mu).
\]
It follows directly from the definitions that
\[
\phi^{\ast}\in\Lambda_{\delta}\left(  P_{n}\right)  \Longleftrightarrow
\mathcal{\bar{R}}_{n}^{\ast}(\phi^{\ast})\leq\delta.
\]
Therefore
\[
\delta_{n}^{\ast}=\inf\{\delta:P^{\ast}(\mathcal{\bar{R}}_{n}^{\ast}%
(\phi^{\ast})\leq\delta)\geq1-\delta_{0}\}.
\]
In other words, $\delta_{n}^{\ast}$ is the quantile corresponding to the
$1-\delta_{0}$ percentile of the distribution of $\mathcal{\bar{R}}_{n}^{\ast
}(\phi^{\ast})$.\footnote{Herein the analysis is under Assumption A3). If
$\Phi_{P^{\ast}}$ contained more than just one element, then there would be
several possible options to formulate an optimization problem for choosing
$\delta$. For example, we may choose $\delta$ as the smallest uncertainty size
such that $\Phi_{P^{\ast}}\subset\Lambda_{\delta}\left(  P_{n}\right)  $ with
probability $1-\delta_{0}$, in which case we would need to study $\sup
_{\phi^{\ast}\in\Phi_{P^{\ast}}}\mathcal{\bar{R}}_{n}^{\ast}(\phi^{\ast})$.}

Still, even under A3), the statistic $\mathcal{\bar{R}}_{n}^{\ast}(\phi^{\ast
})$ is somewhat cumbersome to work with as it is derived from solving a
minimization problem in terms of the mean and variance. So, instead, we will
define an alternative statistic involving only the empirical mean and variance
while producing an upper bound which still preserves the target rate of
convergence to zero as $n\rightarrow\infty$ (which, as we have argued, should
be of order $O\left(  n^{-1}\right)  $).

Denote $\Sigma_{n}=E_{P_{n}}\left(  RR^{T}\right)  $, and let $\lambda
_{1}^{\ast}$ be the Lagrange multiplier in (\ref{LagStar}). Set
\begin{equation}
\mu_{n}=\rho1+2\left(  \Sigma_{n}\phi^{\ast}-\phi^{\ast T}\Sigma_{n}\phi
^{\ast}1\right)  /\lambda_{1}^{\ast}.
\end{equation}

Define
\[
\mathcal{R}_{n}(\Sigma_{n},\mu_{n}):=\mathcal{\bar{R}}_{n}(\phi^{\ast}%
,\lambda_{1}^{\ast},\Sigma_{n},\mu_{n}).
\]
It is clear that
\[
\mathcal{R}_{n}(\Sigma_{n},\mu_{n})\geq\mathcal{\bar{R}}_{n}^{\ast}(\phi
^{\ast}).
\]
Therefore,
\[
\mathcal{R}_{n}(\Sigma_{n},\mu_{n})\leq\delta\Longrightarrow\mathcal{\bar{R}%
}_{n}^{\ast}(\phi^{\ast})\leq\delta\
\]
and, consequently,
\[
\bar{\delta}_{n}^{\ast}=\inf\{\delta\geq0:P^{\ast}\left(  \mathcal{R}%
_{n}(\Sigma_{n},\mu_{n})\leq\delta\right)  \geq1-\delta_{0}\}\geq\delta
_{n}^{\ast}.
\]

Moreover, because of the choice of $\Sigma_{n}$ and $\mu_{n}$,
we have
\[
\mathcal{R}_{n}(\Sigma_{n},\mu_{n})=\inf\{\mathcal{D}_{c}(P,P_{n}%
):E_{P}[RR^{T}]=\Sigma_{n},E_{P}[R]=\mu_{n}\}.
\]
The next result shows $\bar{\delta}_{n}^{\ast}=O\left(n^{-1}\right)  $ as
$n\rightarrow\infty$.

\begin{theorem}
\label{ThmMain2} Assume A1) and A2) hold and write $\mu_{\ast}=E_{P^{\ast}%
}\left(  R\right)  $ and $\Sigma_{\ast}=E_{P^{\ast}}\left(  RR^{T}\right)  $.
Define
\[
g\left(  x\right)  =xx^{T} +2\left(  x \cdot\phi^{\ast}-\phi^{\ast T}x
\phi^{\ast}1\right)  /\lambda_{1}^{\ast}.
\]
Then%
\[
nR_{n}(\Sigma_{n},\mu_{n})\Rightarrow L_{0}:=\sup_{\bar{\lambda}\in
\mathbb{R}^{d}}\left( \bar{\lambda}^{T}Z-\inf_{\bar{\Lambda}\in\mathbb{R}%
^{d\times d}}E_{P^{\ast}}[\left\Vert \bar{\Lambda}R+\bar{\lambda}\right\Vert
_{p}^{2}]\right)
\]
where $Z\sim N\left(  0,\Upsilon_{g}\right)  $. Moreover, if $p=2$ then
\[
L_{0}=\left\Vert Z\right\Vert _{2}^{2}/\left(  1-\left\Vert \mu_{\ast
}\right\Vert _{2}^{4}/\mu_{\ast}^{T}\Sigma_{\ast}\mu_{\ast}\right)  .
\]

\end{theorem}

A proof of Theorem \ref{ThmMain2} is provided in the Appendix.

Note that $L_0$ has an explicit expression when $p=2$. When $p\neq 2$,
using the inequalities that $||x||_{p}^{2}%
\geq||x||_{2}^{2}$ if $p<2$ and $d^{(\frac{1}{2}-\frac{1}{p})}||x||_{p}%
^{2}\geq||x||_{2}^{2}$ if $p>2$, we can find
a stochastic upper
bound of $L_0$ that can be
explicitly expressed. In that case we can obtain $\bar{\delta}_{n}^{\ast}$ in
exactly the same way.

\subsection{Choice of $\bar{\alpha}$\label{Sec_Choice_Alpha}}

Once $\delta$ has been chosen, the next step is to choose $\bar{\alpha}$.
The idea is to select $\bar{\alpha}$ just large enough to make sure that we do
not rule out that $\phi^{\ast}\in\mathcal{F}_{\delta,\bar{\alpha}}\left(
n\right)  $ holds with a given confidence level chosen by the user, where
$\phi^{\ast}$ is the optimal solution to (\ref{Op}). It is equivalent to
choose $\upsilon_{0}$ where
\[
\bar{\alpha}=\rho-\sqrt{\delta}\left\Vert \phi^{\ast}\right\Vert _{{p}%
}\upsilon_{0}.\text{ }%
\]

Therefore, it follows from Proposition \ref{prop1} that $\phi^{\ast}\in
\mathcal{F}_{\delta,\bar{\alpha}}\left(  n\right)  $ if and only if
\[
\left(  \phi^{\ast}\right)  ^{T}E_{P_{n}}\left(  R\right)  -\sqrt{\delta
}\left\Vert \phi^{\ast}\right\Vert _{{p}}\geq\rho-\sqrt{\delta}\left\Vert
\phi^{\ast}\right\Vert _{{p}}\upsilon_{0}.
\]
However, $\rho=\left(  \phi^{\ast}\right)  ^{T}E_{P^{\ast}}\left(  R\right)
$; so the previous inequality holds if and only if%
\begin{equation}
\left(  \phi^{\ast}\right)  ^{T}\left(  E_{P_{n}}\left(  R\right)
-E_{P^{\ast}}\left(  R\right)  \right)  \geq \left\Vert \phi^{\ast}\right\Vert
_{{p}}\sqrt{\delta}\left(1-  \upsilon_{0}\right)  .\label{Bnd_2abc}%
\end{equation}
Therefore, we can choose $\sqrt{\delta}\left(1-  \upsilon_{0}\right)  <0$
sufficiently negative so that the previous inequality holds with a specified
confidence level. We hope to choose a $v_{0}$ such that $\phi^{\ast}$ will
satisfy (\ref{Bnd_2abc}) with confidence level $1-\epsilon$. This can be
achieved asymptotically by the central limit theorem as the following
result indicates.

\begin{proposition}
\label{Prop_Choice_v0}Suppose that A1) and A3) hold and let $\left\{  \phi
_{n}^{\ast}\right\}  _{n=1}^{\infty}$ be any consistent sequence of estimators
of $\phi^{\ast}$ in the sense that $\phi_{n}^{\ast}\rightarrow\phi^{\ast}$ in
probability as $n\rightarrow\infty$. Then,
\[
n^{1/2}\left[\frac{\left(  \phi_{n}^{\ast}\right)  ^{T}\left(  E_{P_{n}}\left(  R\right)
-E_{P^{\ast}}\left(  R\right)  \right)  }{\left\Vert \phi_{n}^{\ast
}\right\Vert _{p}}\right]\Rightarrow N\left(  0,\Upsilon_{\phi^{\ast}}\right),
\]
as $n\rightarrow\infty$, where
\[
\Upsilon_{\phi^{\ast}}:=\lim_{n\rightarrow\infty}Var_{P^{\ast}}\left(  n^{-1/2}\sum_{k=1}^{n}\left(
\phi^{\ast}\right)  ^{T}R_{k}/\left\Vert \phi^{\ast}\right\Vert _{{p}}\right)
.
\]

\end{proposition}

\begin{proof}
Note that%
\begin{align*}
n^{1/2}\left[\frac{\left(  \phi_{n}^{\ast}\right)  ^{T}\left(  E_{P_{n}}\left(  R\right)
-E_{P^{\ast}}\left(  R\right)  \right)  }{\left\Vert \phi_{n}^{\ast
}\right\Vert _{p}}\right]  & =n^{1/2}\left[\frac{\left(  \phi_{n}^{\ast}-\phi^{\ast}\right)
^{T}\left(  E_{P_{n}}\left(  R\right)  -E_{P^{\ast}}\left(  R\right)  \right)
}{\left\Vert \phi_{n}^{\ast}\right\Vert _{p}}\right]\\
& +n^{1/2}\left[\frac{\left(  \phi^{\ast}\right)  ^{T}\left(  E_{P_{n}}\left(  R\right)
-E_{P^{\ast}}\left(  R\right)  \right)  }{\left\Vert \phi^{\ast}\right\Vert
_{p}}\cdot\frac{\left\Vert \phi^{\ast}\right\Vert _{p}}{\left\Vert \phi
_{n}^{\ast}\right\Vert _{p}}\right].
\end{align*}
By the central limit theorem and the fact that $\phi_{n}^{\ast}\rightarrow0$ in
probability, we conclude
\[
n^{1/2}\left[\frac{\left(  \phi_{n}^{\ast}\right)  ^{T}\left(  E_{P_{n}}\left(  R\right)
-E_{P^{\ast}}\left(  R\right)  \right)  }{\left\Vert \phi_{n}^{\ast
}\right\Vert _{p}}-\frac{\left(  \phi^{\ast}\right)  ^{T}\left(  E_{P_{n}%
}\left(  R\right)  -E_{P^{\ast}}\left(  R\right)  \right)  }{\left\Vert
\phi^{\ast}\right\Vert _{p}}\right]\Rightarrow0
\]
as $n\rightarrow\infty$. However,
\[
n^{1/2}\left[\frac{\left(  \phi^{\ast}\right)  ^{T}\left(  E_{P_{n}}\left(  R\right)
-E_{P^{\ast}}\left(  R\right)  \right)  }{\left\Vert \phi^{\ast}\right\Vert
_{p}}\right]\Rightarrow N\left(  0,\Upsilon_{\phi^{\ast}}\right) ,
\]
which implies the desired result.
\end{proof}

Using the previous result we can estimate $v_{0}$ asymptotically. Let
$\phi_{n}$ denote the optimal solution of problem $\mathcal{Q}(P_{n})$. We
know that $\phi_{n}$ converges to $\phi^{\ast}$ in probability. So, in
practice we choose a $v_{0}$ such that the following inequality will hold with
confidence level $1-\epsilon$,
\begin{equation}
\frac{1}{||\phi_{n}||_{p}}(\phi_{n})^{T}(E_{P_{n}}(R)-E_{P^{\ast}}%
(R))\geq\sqrt{\delta}(1-v_{0}).\label{eqv0}%
\end{equation}
The left-hand side of (\ref{eqv0}) is approximately normally distributed and
thus we can choose its $1-\epsilon$ quantile and then we are be able to decide
the value of $v_{0}>1$.

\section{Concluding Remarks\label{Sec_Conc}}

We have provided a data-driven DRO theory for Markowitz's mean--variance
portfolio selection. The robust model can be solved via a non-robust one based
on the empirical probability measure with an additional regularization term.
The size of the distributional uncertainty region is not exogenously given;
rather it is informed by the return data in a scheme which we have developed
in this paper.

Our results may be generalized in different directions. We have chosen the
$l_{q}$\ norm in defining our Wasserstein distance due to its popularity in
regularization, but other transportation costs can be used. For example, one
may consider the type of transportation cost related to adaptive
regularization that has been studied by Blanchet, Kang, Zhang and Murthy
(2017), or the one related to industry cluster as in Blanchet and Kang (2017).
Another significant direction is a dynamic (discrete-time or continuous-time)
version of the DRO Markowitz model.

\bigskip

\bigskip

\noindent\textbf{{\Large Appendices}}

\appendix

\section{Proof of Proposition 1\label{prop1proof}}

We consider the following problem
\begin{equation}
\min\limits_{P\in D_{c}(P,P_{n})\leq\delta}\phi^{T}E_{P}[R]
\end{equation}
or, equivalently,
\begin{equation}
-\max\limits_{P\in D_{c}(P,P_{n})\leq\delta}E_{P}[(-\phi)^{T}R].
\end{equation}
By checking Slater's condition and using Proposition 4 of Blanchet, Kang and
Murthy (2016) we obtain the dual problem:
\begin{equation}
\max\limits_{P\in D_{c}(P,P_{n})\leq\delta}E_{P}[(-\phi)^{T}R]=\inf
\limits_{\lambda\geq0}\left[ \lambda\delta+\frac{1}{n}\Phi_{\lambda}%
(R_{i})\right]  \label{dual1}%
\end{equation}
where
\begin{align*}
\Phi_{\lambda}(R_{i})  &  =\sup\limits_{u}\{h(u)-\lambda c(u,R_{i})\}\\
&  =\sup\limits_{u}\{(-\phi^{T})u-\lambda||u-R_{i}||_{q}^{2}\}\\
&  =\sup\limits_{\Delta}\{(-\phi^{T})(\Delta+R_{i})-\lambda||\Delta||_{q}%
^{2}\}\\
&  =\sup\limits_{\Delta}\{(-\phi^{T})\Delta-\lambda||\Delta||_{q}^{2}%
\}-\phi^{T}R_{i}\\
&  =\sup\limits_{\Delta}\{||\phi||_{p}||\Delta||_{q}-\lambda||\Delta||_{q}%
^{2}\}-\phi^{T}R_{i}\\
&  =\frac{||\phi||_{p}^{2}}{4\lambda}-\phi^{T}R_{i}.
\end{align*}
Thus, (\ref{dual1}) becomes
\begin{align*}
\max\limits_{P\in D_{c}(P,P_{n})\leq\delta}E_{P}[(-\phi)^{T}R]  &
=\inf\limits_{\lambda\geq0}\{\lambda\delta+\frac{1}{n}[\frac{||\phi||_{p}^{2}%
}{4\lambda}-\phi^{T}R_{i}]\}\\
&  =\inf\limits_{\lambda\geq0}\{\lambda\delta+\frac{||\phi||_{p}^{2}}%
{4\lambda}-\phi^{T}E_{P_{n}}[R]\}\\
&  =\sqrt{\delta}||\phi||_{p}-\phi^{T}E_{P_{n}}[R]
\end{align*}
or
\begin{equation}
\min\limits_{P\in D_{c}(P,P_{n})\leq\delta}\phi^{T}E_{P}[R]=\phi^{T}E_{P_{n}%
}[R]-\sqrt{\delta}||\phi||_{p}.
\end{equation}

\section{Proof of Proposition 2\label{prop2}}

Introducing a slack random variable $S\equiv v$, where $v$ is a deterministic
number. Then we can recast problem (\ref{maxpp}) as
\begin{align}
\label{maxp}\max\{E_{P}[(U^{T}\phi)^{2}]  &  :E_{\pi}[c(U,R)+S]=\delta,\pi
_{R}=P_{n},\pi(S=v)=1,\\
E_{\pi}[U^{T}\phi]  &  =\alpha,\pi\in\mathcal{P}(\mathcal{R}^{m}%
\times\mathcal{R}^{m}\times\mathcal{R}_{+})\}.
\end{align}
Define
\[
\Omega:=\{(u,r,s):c(u,r)<\infty,s\geq0\},
\]
and let

\begin{center}%
\begin{equation}
f\left(  u,r,s\right)  =\left[
\begin{matrix}
1_{r=R_{1}}(u,r,s)\\
...\\
1_{r=R_{n}}(u,r,s)\\
\phi^{T}u\\
1_{s=v}(u,r,s)\\
c(u,r)+s
\end{matrix}
\right]  \ \ \mathrm{and}\ \ q=\left[
\begin{matrix}
\frac{1}{n}\\
...\\
\frac{1}{n}\\
\alpha\\
1\\
\delta
\end{matrix}
\right] .
\end{equation}

\end{center}

Thus (\ref{maxp}) can be written as,
\begin{equation}
\max\{E_{\pi}[(U^{T}\phi)^{2}]:E_{\pi}[f(U,R,S)]=q,\pi\in\mathcal{P}_{\Omega
}\}. \label{maxp2}%
\end{equation}
Let $f_{0}=\mathbf{1}_{\Omega}$, $\tilde{f}=(f_{0},f)$, $\tilde{q}=(1,q)$,
$\mathcal{Q}_{\tilde{f}}:=\{\int\tilde{f}(x)d\mu(x):\mu\in\mathcal{M}_{\Omega
}^{+}\}$ where $\mathcal{M}_{\Omega}^{+}$ denote the set of non-negative
measures on $\Omega$. If $\phi\neq0$, then it is easy to see that $\tilde{q}$
lies in the interior of $\mathcal{Q}_{\tilde{f}}$. By Proposition 6 in
Blanchet, Kang and Murthy (2016), the optimal value of problem (\ref{maxp2})
equals to that of its dual problem, i.e.,
\begin{align}
\max\{E_{\pi}[(U^{T}\phi)^{2}]  &  :E_{\pi}[f(U,R,S)]=q,\pi\in\mathcal{P}%
_{\Omega}\}\label{maxp3}\\
&  =\inf\limits_{a=\left(  a_{0},...,a_{n}\right)  \in A}\{a_{0}+\frac{1}%
{n}\sum\limits_{i=1}^{n}a_{i}+\alpha a_{n+1}+a_{n+2}+\delta a_{n+3}%
\},\nonumber
\end{align}
where
\begin{align*}
A:=  &  \{a=\left(  a_{0},...,a_{n}\right)  :a_{0}+\frac{1}{n}\sum
\limits_{i=1}^{n}a_{i}1_{r=R_{i}}(u,r,s)+a_{n+1}\phi^{T}u\\
&  \left.  +a_{n+2}1_{s=v}(u,r,s)+a_{n+3}[c(u,r)+s]\geq(\phi^{T}u)^{2}%
,\forall(u,r,s)\in\Omega\}.\right.
\end{align*}
From the definition of $A$, replacing $r=R_{i}$, we obtain that the inequality%
\begin{equation}
a_{0}+a_{i}+a_{n+2}\geq\sup\limits_{(u,s)\in\Omega}\{(\phi^{T}u)^{2}%
-a_{n+3}[c(u,R_{i})+s]-a_{n+1}\phi^{T}u\}
\end{equation}
holds for each $i\in\{1,...,n\}$. It follows directly that
\begin{align}
&  \sup\limits_{(u,s)\in\Omega}\{(\phi^{T}u)^{2}-a_{n+3}[c(u,R_{i}%
)+s]-a_{n+1}\phi^{T}u\}\\
&  =%
\begin{cases}
+\infty, & \mbox{ if $a_{n+3}<0$}\\
\sup\limits_{u}\{(\phi^{T}u)^{2}-a_{n+3}c(u,R_{i})-a_{n+1}\phi^{T}u\}, &
\mbox{ if $a_{n+3}\geq0$}.
\end{cases}
\end{align}
Thus, the dual problem can be expressed as
\begin{equation}
\inf\{a_{0}+\frac{1}{n}\sum\limits_{i=1}^{n}a_{i}+\alpha a_{n+1}%
+a_{n+2}+\delta a_{n+3}:\\
a_{n+3}\geq0,a_{0}+a_{i}+a_{n+2}\geq\sup\limits_{u}\{(\phi^{T}u)^{2}%
-a_{n+3}c(u,R_{i})-a_{n+1}\phi^{T}u\}\}
\end{equation}
which can be transformed into
\begin{equation}
\inf\limits_{a_{n+3}\geq0}\{\frac{1}{n}\sum\limits_{i=1}^{n}\Phi(R_{i})+\alpha
a_{n+1}+\delta a_{n+3}\},
\end{equation}
with
\[
\Phi(R_{i}):=\sup\limits_{u}\{(\phi^{T}u)^{2}-a_{n+3}c(u,R_{i})-a_{n+1}%
\phi^{T}u\}.
\]
Using $\lambda_{1}$ to replace $a_{n+3}$ and $\lambda_{2}$ to replace
$a_{n+1}$, the dual problem becomes
\begin{equation}
\inf\limits_{\lambda_{1}\geq0}\{\frac{1}{n}\sum\limits_{i=1}^{n}\Phi
(R_{i})+\lambda_{2}\alpha+\lambda_{1}\delta\}
\end{equation}
where
\[
\Phi(R_{i}):=\sup\limits_{u}\{(\phi^{T}u)^{2}-\lambda_{1}c(u,R_{i}%
)-\lambda_{2}\phi^{T}u\}.
\]

\section{Proof of Proposition 3\label{prop3}}

Writing $\Delta:=u-R_{i}$, we have
\begin{align*}
\Phi(R_{i})  &  =\sup\limits_{u}\{(\phi^{T}u)^{2}-\lambda_{1}c(u,R_{i}%
)-\lambda_{2}\phi^{T}u\}\\
&  =\sup\limits_{u}\{(\phi^{T}u)^{2}-\lambda_{1}||u-R_{i}||_{q}^{2}%
-\lambda_{2}\phi^{T}u\}\\
&  =\sup\limits_{\Delta}\{(\phi^{T}(\Delta+R_{i}))^{2}-\lambda_{1}%
||\Delta||_{q}^{2}-\lambda_{2}\phi^{T}(R_{i}+\Delta)\}\\
&  =\sup\limits_{\Delta}\{(\phi^{T}R_{i})^{2}+(\phi^{T}\Delta)^{2}+2(\phi
^{T}R_{i})(\phi^{T}\Delta)-\lambda_{1}||\Delta||_{q}^{2}-\lambda_{2}\phi
^{T}(R_{i}+\Delta)\}\\
&  =(\phi^{T}R_{i})^{2}-\lambda_{2}\phi^{T}R_{i}+\sup\limits_{\Delta}%
\{(\phi^{T}\Delta)^{2}+2(\phi^{T}R_{i})(\phi^{T}\Delta)-\lambda_{1}%
||\Delta||_{q}^{2}-\lambda_{2}\phi^{T}\Delta\}\\
&  =(\phi^{T}R_{i})^{2}-\lambda_{2}\phi^{T}R_{i}+\sup\limits_{\Delta}%
\{(||\phi||_{p}^{2}-\lambda_{1})||\Delta||_{q}^{2}+|2(R_{i}^{T}\phi
)-\lambda_{2}|(||\phi||_{p}||\Delta||_{q})\}.
\end{align*}
We can consider four cases: 1) $||\phi||_{p}^{2}>\lambda_{1}$, $\Phi
(R_{i})=+\infty$; 2) $||\phi||_{p}^{2}=\lambda_{1}$, $2R_{i}^{T}\phi
\neq\lambda_{2}$, $\Phi(R_{i})=+\infty$; 3) $||\phi||_{p}^{2}=\lambda_{1}$,
$2R_{i}^{T}\phi=\lambda_{2}$, $\Phi(R_{i})=0$; 4) $||\phi||_{p}^{2}%
<\lambda_{1}$, $\Phi(R_{i})=(\phi^{T}R_{i})^{2}-\lambda_{2}\phi^{T}R_{i}%
+\frac{(2R_{i}^{T}\phi-\lambda_{2})^{2}||\phi||_{p}^{2}}{4(\lambda_{1}%
-||\phi||_{p}^{2})}$.

For any of the first three cases, the value of $\frac{1}{n}\sum\limits_{i=1}%
^{n}\Phi(R_{i})$ is $+\infty$.
Hence only the fourth case is non-trivial. In this case, problem
(\ref{dualpp}) is transformed into
\begin{equation}%
\begin{split}
&  \inf\limits_{\lambda_{1}\geq0,\lambda_{2}}[\frac{1}{n}\sum\limits_{i=1}%
^{n}\Phi(R_{i})+\lambda_{2}\alpha+\lambda_{1}\delta]\\
&  =\inf\limits_{\lambda_{1}\geq||\phi||_{p}^{2},\lambda_{2}}\left\{ \frac
{1}{n}\sum\limits_{i=1}^{n}\left[ (\phi^{T}R_{i})^{2}-\lambda_{2}\phi^{T}%
R_{i}+\frac{(2R_{i}^{T}\phi-\lambda_{2})^{2}||\phi||_{p}^{2}}{4(\lambda
_{1}-||\phi||_{p}^{2})}\right] +\lambda_{2}\alpha+\lambda_{1}\delta\right\}
\end{split}
\label{p2}%
\end{equation}
Define%
\[
H=\frac{1}{n}\sum_{i=1}^{n}\left[ \left(  \phi^{T}R_{i}\right)  ^{2}%
-\lambda_{2}\phi^{T}R_{i}+\frac{\left(  2R_{i}^{T}\phi-\lambda_{2}\right)
^{2}\left\Vert \phi\right\Vert _{p}^{2}}{4\left(  \lambda_{1}-\left\Vert
\phi\right\Vert _{p}^{2}\right)  }\right] +\lambda_{2}\alpha+\lambda_{1}%
\delta.
\]
Taking partial derivative with respect to $\lambda_{2}$ and setting it to be
$0$, we get%

\[
\frac{\partial H}{\partial\lambda_{2}}=\alpha-\frac{1}{n}\sum\limits_{i=1}%
^{n}\left[ \phi^{T}R_{i}+\frac{(2\phi^{T}R_{i}-\lambda_{2})||\phi||_{p}^{2}%
}{2(\lambda_{1}-||\phi||_{p}^{2})}\right] =0
\]
which implies (note that $\phi^{T}1=1$ guarantees that $||\phi||_{p}^{2}>0$)%
\begin{equation}
\lambda_{2}=2\alpha-2C\frac{\lambda_{1}}{||\phi||_{p}^{2}} \label{lamb2}%
\end{equation}
where $C:=\alpha-\phi^{T}E_{P_{n}}[R]$.
Moreover, $\lambda_{2}$ is optimal because%
\begin{equation}
\frac{\partial^{2}H}{\partial\lambda_{2}^{2}}=\frac{||\phi||_{p}^{2}%
}{2(\lambda_{1}-||\phi||_{p}^{2})}>0.
\end{equation}
We plug (\ref{lamb2}) into (\ref{p2}) and obtain
\begin{align*}
&  \inf\limits_{\lambda_{1}\geq0,\lambda_{2}}\left[ \frac{1}{n}\sum
\limits_{i=1}^{n}\Phi(R_{i})+\lambda_{2}\alpha+\lambda_{1}\delta\right] \\
&  =\frac{1}{n}\sum\limits_{i=1}^{n}(\phi^{T}R_{i})^{2}+\inf\limits_{\lambda
_{1}\geq||\phi||_{p}^{2},\lambda_{2}}\left\{ \frac{1}{n}\sum\limits_{i=1}%
^{n}\left[ -\lambda_{2}\phi^{T}R_{i}+\frac{(2R_{i}^{T}\phi-\lambda_{2}%
)^{2}||\phi||_{p}^{2}}{4(\lambda_{1}-||\phi||_{p}^{2})}\right] +\lambda
_{2}\alpha+\lambda_{1}\delta\right\} \\
&  =\frac{1}{n}\sum\limits_{i=1}^{n}(\phi^{T}R_{i})^{2}+\inf\limits_{\lambda
_{1}\geq||\phi||_{p}^{2},\lambda_{2}}\left\{ \frac{1}{n}\sum\limits_{i=1}%
^{n}\left[ \frac{(2R_{i}^{T}\phi-\lambda_{2})^{2}||\phi||_{p}^{2}}%
{4(\lambda_{1}-||\phi||_{p}^{2})}\right] +\lambda_{2}C+\lambda_{1}%
\delta\right\} \\
&  =\frac{1}{n}\sum\limits_{i=1}^{n}(\phi^{T}R_{i})^{2}+\inf\limits_{\lambda
_{1}\geq||\phi||_{p}^{2}}\left\{ \frac{1}{n}\sum\limits_{i=1}^{n}\left[
\frac{(2R_{i}^{T}\phi-2\alpha+2C\frac{\lambda_{1}}{||\phi||_{p}^{2}}%
)^{2}||\phi||_{p}^{2}}{4(\lambda_{1}-||\phi||_{p}^{2})}\right] +(2\alpha
-2C\frac{\lambda_{1}}{||\phi||_{p}^{2}})C+\lambda_{1}\delta\right\} .
\end{align*}
Writing $\lambda_{1}=\kappa+||\phi||_{p}^{2}$, we have
\begin{align*}
&  \inf\limits_{\lambda_{1}\geq0,\lambda_{2}}\left[ \frac{1}{n}\sum
\limits_{i=1}^{n}\Phi(R_{i})+\lambda_{2}\alpha+\lambda_{1}\delta\right] \\
&  =\frac{1}{n}\sum\limits_{i=1}^{n}(\phi^{T}R_{i})^{2}+\inf\limits_{\kappa
\geq0}\left\{ \frac{1}{n}\sum\limits_{i=1}^{n}\left[ \frac{(R_{i}^{T}%
\phi-\alpha+C\frac{k+|\phi||_{p}^{2}}{||\phi||_{p}^{2}})^{2}N}{\kappa}\right]
+(2\alpha-2C\frac{\kappa+||\phi||_{p}^{2}}{||\phi||_{p}^{2}})C+(\kappa
+||\phi||_{p}^{2})\delta\right\} \\
&  =\frac{1}{n}\sum\limits_{i=1}^{n}(\phi^{T}R_{i})^{2}+\inf\limits_{\kappa
\geq0}\{\frac{C_{1}^{2}}{||\phi||_{p}^{2}}k+2||\phi||_{p}^{2}(\phi^{T}\bar
{W}-\alpha+C)+\frac{1}{n}\sum\limits_{i=1}^{n}\frac{(R_{i}^{T}\phi
-\alpha+C)^{2}||\phi||_{p}^{2}}{\kappa}\\
&  +2\alpha C-2C^{2}+\kappa(\delta-\frac{2C^{2}}{||\phi||_{p}^{2}}%
)+||\phi||_{p}^{2}\delta\}\\
&  =\frac{1}{n}\sum\limits_{i=1}^{n}(\phi^{T}R_{i})^{2}+2\alpha C-2C^{2}%
+|\phi||_{p}^{2}\delta+\inf\limits_{\kappa\geq0}\left\{ \frac{1}{n}%
\sum\limits_{i=1}^{n}\frac{(R_{i}^{T}\phi-E_{P_{n}}[R]\phi)^{2}||\phi
||_{p}^{2}}{\kappa}+\kappa(\delta-\frac{C^{2}}{||\phi||_{p}^{2}})\right\} .
\end{align*}
If $\delta-C^{2}/||\phi||_{p}^{2}<0$, then the optimal value of the above
problem is $-\infty$, which means that the primal problem (\ref{maxpp}) is not
feasible. If $\delta-C^{2}/||\phi||_{p}^{2}\geq0$, then
\begin{align*}
&  \frac{1}{n}\sum\limits_{i=1}^{n}(\phi^{T}R_{i})^{2}+2\alpha C-2C^{2}%
+|\phi||_{p}^{2}\delta+\inf\limits_{\kappa\geq0}\left\{ \frac{1}{n}%
\sum\limits_{i=1}^{n}\frac{(R_{i}^{T}\phi-E_{P_{n}}[R]\phi)^{2}||\phi
||_{p}^{2}}{\kappa}+\kappa(\delta-\frac{C^{2}}{||\phi||_{p}^{2}})\right\} \\
&  =\frac{1}{n}\sum\limits_{i=1}^{n}(\phi^{T}R_{i})^{2}+2(\alpha-\phi
^{T}E_{P_{n}}[R])\phi^{T}E_{P_{n}}[R]+\delta||\phi||_{p}^{2}\\
&  +2\sqrt{\delta||\phi||_{p}^{2}-(\alpha-\phi^{T}E_{P_{n}}[R])^{2}}%
\sqrt{\frac{1}{n}\phi^{T}\sum\limits_{i=1}^{n}(R_{i}-E_{P_{n}}[R])(R_{i}%
-E_{P_{n}}[R])^{T}\phi}\\
&  =\frac{1}{n}\sum\limits_{i=1}^{n}(\phi^{T}R_{i})^{2}+2(\alpha-\phi
^{T}E_{P_{n}}[R])\phi^{T}E_{P_{n}}[R]+\delta||\phi||_{p}^{2}\\
&  +2\sqrt{\delta||\phi||_{p}^{2}-(\alpha-\phi^{T}E_{P_{n}}[R])^{2}}\sqrt
{\phi^{T}Var_{P_{n}}[R]\phi}.
\end{align*}
Thus, problem (\ref{maxpp}) can be written as
\begin{align*}
&  \min\limits_{\phi}\ \ \frac{1}{n}\sum\limits_{i=1}^{n}(\phi^{T}R_{i}%
)^{2}+2(\alpha-\phi^{T}E_{P_{n}}[R])\phi^{T}E_{P_{n}}[R]+\delta||\phi
||_{p}^{2}\\
&  +2\sqrt{\delta||\phi||_{p}^{2}-(\alpha-\phi^{T}E_{P_{n}}[R])^{2}}\sqrt
{\phi^{T}Var_{P_{n}}[R]\phi},
\end{align*}
subject to $1^{T}\phi=1$ and $(\alpha-\phi^{T}E_{P_{n}}[R])^{2}-\delta
||\phi||_{p}^{2}\leq0$.

\section{Proof of Theorem 2\label{theorem2}}

Define
\[
h_{0}\left(  R,\Sigma\right)  =RR^{T}-\Sigma\text{ \ and \ }h_{1}\left(
R,\mu\right)  =R-\mu.
\]
Then, by Proposition 1 of Blanchet, Kang and Murthy (2016) we have that for
any given $\mu$ and $\Sigma$,%
\[
\mathcal{R}_{n}(\Sigma,\mu)=\sup\limits_{\Lambda\in\mathbb{R}^{d\times
d},\lambda\in\mathbb{R}^{d}}\left\{ -E_{P_{n}}[\sup\limits_{u\in\mathbb{R}%
^{d}}\{Tr\left(  \Lambda h_{0}(u,\Sigma)\right)  +\lambda^{T}h_{1}\left(
u,\mu\right)  -||u-R||_{q}^{2}\}]\right\} .
\]
Observe that
\begin{align*}
&  \sup\limits_{u\in\mathbb{R}^{d}}\{Tr\left(  \Lambda h_{0}(u,\Sigma)\right)
+\lambda^{T}h_{1}\left(  u,\mu\right)  -||u-R||_{q}^{2}\}\\
&  =\sup\limits_{\Delta\in\mathbb{R}^{d}}\{Tr\left(  \Lambda h_{0}%
(\Delta+R,\Sigma)\right)  +\lambda^{T}h_{1}\left(  \Delta+R,\mu\right)
-||\Delta||_{q}^{2}\}\\
&  =\sup\limits_{\Delta\in\mathbb{R}^{d}}\{Tr\left(  \Lambda\left[
h_{0}(\Delta+R,\Sigma)-h_{0}(R,\Sigma)\right]  \right)  +\lambda^{T}%
\Delta-||\Delta||_{q}^{2}\}\\
&  +Tr(\Lambda h_{0}(R,\Sigma))+\lambda^{T}h_{1}\left(  R,\mu\right) .
\end{align*}

Moreover, let us write
\[
Tr\left(  \Lambda\left[  h_{0}(\Delta+R,\Sigma)-h_{0}(R,\Sigma)\right]
\right)  =\int_{0}^{1}\frac{d}{dt}Tr\left(  \Lambda h_{0}\left(
R+t\Delta\right)  \right)  dt.
\]
However,
\begin{align*}
\frac{d}{dt}Tr\left(  \Lambda h_{0}\left(  R+t\Delta\right)  \right)   &
=2Tr\left(  \Lambda\left(  R+t\Delta\right)  \Delta^{T}\right) \\
&  =2Tr\left(  \Lambda R\Delta^{T}\right)  +2t\Delta^{T}\Lambda\Delta.
\end{align*}
Furthermore,%
\begin{equation}
\left.  E_{P_{n}}\left[  Tr(\Lambda h_{0}(R,\Sigma))\right]  \right\vert
_{\Sigma=\Sigma_{n}}=0. \label{1st_Term}%
\end{equation}
So, we deduce
\begin{align*}
&  \mathcal{R}_{n}(\Sigma_{n},\mu)\\
&  =\sup_{\lambda\in\mathbb{R}^{d}}\{-E_{P_{n}}[\lambda^{T}\left(
R-\mu\right)  ]+\\
&  \sup\limits_{\Lambda\in\mathbb{R}^{d\times d}}(-E_{P_{n}}[\sup_{\Delta
}\{2Tr\left(  \Lambda R\Delta^{T}\right)  +\Delta^{T}\Lambda\Delta+\lambda
^{T}\Delta-||\Delta||_{q}^{2}\}])\}.
\end{align*}
Introduce the scaling $\Delta=\bar{\Delta}/n^{1/2}$ and $\bar{\lambda}=\lambda
n^{1/2}$ and $\bar{\Lambda}=\Lambda n^{1/2}$. Then we obtain
\begin{align*}
&  n\mathcal{R}_{n}(\Sigma_{n},\mu_{n})\\
&  =\sup_{\bar{\lambda}\in\mathbb{R}^{d}}\{-n^{-1/2}\sum_{i=1}^{n}\bar
{\lambda}^{T}\left(  R_{i}-\mu_{n}\right)  +\\
&  \sup\limits_{\bar{\Lambda}\in\mathbb{R}^{d\times d}}(-E_{P_{n}}[\sup
_{\bar{\Delta}}\{2Tr\left(  \bar{\Lambda}R\bar{\Delta}^{T}\right)
+\bar{\Delta}^{T}\bar{\Lambda}\bar{\Delta}/n^{1/2}+\bar{\lambda}^{T}%
\bar{\Delta}-||\bar{\Delta}||_{q}^{2}\}])\}.
\end{align*}
In the proof of Proposition 3 in Blanchet, Kang and Murthy (2016), under
Assumption A2), a technique is introduced to show that $\bar{\Delta}$ and
$\bar{\lambda}$ can be restricted to compact sets with high probability and
therefore the term $\bar{\Delta}^{T}\bar{\Lambda}\bar{\Delta}/n^{1/2}$ is
asymptotically negligible. On the other hand,
\begin{align*}
&  \sup_{\Delta}\{2Tr\left(  \bar{\Delta}^{T}\bar{\Lambda}R\right)
+\bar{\Delta}^{T}\bar{\lambda}-||\bar{\Delta}||_{q}^{2}\}\\
&  =\sup_{\Delta}\{2\left\Vert \bar{\Lambda}R+\bar{\lambda}\right\Vert
_{p}\left\Vert \bar{\Delta}\right\Vert _{q}-||\bar{\Delta}||_{q}%
^{2}\}=\left\Vert \bar{\Lambda}R+\bar{\lambda}\right\Vert _{p}^{2}.
\end{align*}
Therefore, if
\[
n^{-1/2}\sum_{i=1}^{n}\left(  R_{i}-\mu_{n}\right)  \Rightarrow-Z
\]
for some $Z$ (to be characterized momentarily), then we conclude that%
\[
\mathcal{R}_{n}(\Sigma_{n},\mu_{n})\Rightarrow L_{0}=\sup_{\bar{\lambda}%
\in\mathbb{R}^{d}}\{\bar{\lambda}^{T}Z-\inf_{\bar{\Lambda}\in\mathbb{R}%
^{d\times d}}E_{P^{\ast}}[\left\Vert \bar{\Lambda}R+\bar{\lambda}\right\Vert
_{p}^{2}]\}.
\]
If $p=2$ then we have
\[
E_{P^{\ast}}[\left\Vert \bar{\Lambda}R+\bar{\lambda}\right\Vert _{2}^{2}%
]=\sum_{i}E_{P^{\ast}}\left(  \bar{\Lambda}_{i\cdot}\cdot R+\bar{\lambda}%
_{i}\right)  ^{2}.
\]
So, taking derivative with respect to the $i$-th row, $\bar{\Lambda}_{i\cdot}%
$, of the matrix $\bar{\Lambda}$, $\bar{\Lambda}_{i\cdot}$, we obtain
\begin{equation}
\nabla_{\bar{\Lambda}_{i\cdot}}E_{P^{\ast}}[\left\Vert \bar{\Lambda}%
R+\bar{\lambda}\right\Vert _{2}^{2}]=2E_{P^{\ast}}\left(  \left(  R^{T}%
\bar{\Lambda}_{i\cdot}+\bar{\lambda}_{i}\right)  R\right)  =2E_{P^{\ast}%
}\left(  R^{T}\bar{\Lambda}_{i\cdot}R\right)  +2\bar{\lambda}_{i}E_{P^{\ast}%
}\left(  R\right)  =0. \label{ee}%
\end{equation}
Writing
\[
\mu_{\ast}=E_{P^{\ast}}\left(  R\right)  \text{ and }\Sigma_{\ast}=E_{P^{\ast
}}\left(  RR^{T}\right) ,
\]
and then multiplying (\ref{ee}) by $\bar{\Lambda}_{i\cdot}^{T}$ we obtain
\[
\bar{\Lambda}_{i\cdot}^{T}\Sigma_{\ast}\bar{\Lambda}_{i\cdot}=-\bar{\lambda
}_{i}\bar{\Lambda}_{i\cdot}^{T}\mu_{\ast}.
\]
To solve this equation take%
\[
\bar{\Lambda}_{i\cdot}=a_{i}\mu_{\ast},
\]
leading to
\[
a_{i}\mu_{\ast}^{T}\Sigma_{\ast}\mu_{\ast}=-\bar{\lambda}_{i}\left\Vert
\mu_{\ast}\right\Vert _{2}^{2},
\]
or
\[
a_{i}=-\bar{\lambda}_{i}\left\Vert \mu_{\ast}\right\Vert _{2}^{2}/\mu_{\ast
}^{T}\Sigma_{\ast}\mu_{\ast}.
\]
Therefore,%
\begin{align*}
E_{P^{\ast}}\left(  \bar{\Lambda}_{i\cdot}R+\bar{\lambda}_{i}\right)  ^{2}  &
=\bar{\Lambda}_{i\cdot}^{T}\Sigma_{\ast}\bar{\Lambda}_{i\cdot}+2\bar{\lambda
}_{i}\bar{\Lambda}_{i\cdot}^{T}\mu_{\ast}+\bar{\lambda}_{i}^{2}\\
&  =\bar{\lambda}_{i}^{2}+\bar{\lambda}_{i}\bar{\Lambda}_{i\cdot}^{T}\mu
_{\ast}=\bar{\lambda}_{i}^{2}\left(  1-\left\Vert \mu_{\ast}\right\Vert
_{2}^{4}/\mu_{\ast}^{T}\Sigma_{\ast}\mu_{\ast}\right)  .
\end{align*}
Observe that $\left\Vert \mu_{\ast}\right\Vert _{2}^{4}/\mu_{\ast}^{T}%
\Sigma_{\ast}\mu_{\ast}<1$ if and only if
\[
Tr\left(  E_{P^{\ast}}\left(  RR^{T}\right)  E_{P^{\ast}}\left(  R\right)
E_{P^{\ast}}\left(  R^{T}\right)  \right)  >Tr\left(  E_{P^{\ast}}\left(
R\right)  E_{P^{\ast}}\left(  R^{T}\right)  E_{P^{\ast}}\left(  R\right)
E_{P^{\ast}}\left(  R^{T}\right)  \right)  ,
\]
which in turn holds if and only if%
\begin{align*}
&  Tr\left(  E_{P^{\ast}}\left(  R^{T}\right)  \left[  E_{P^{\ast}}\left(
RR^{T}\right)  -E_{P^{\ast}}\left(  R\right)  E_{P^{\ast}}\left(
R^{T}\right)  \right]  E_{P^{\ast}}\left(  R\right)  \right) \\
&  =Var_{P^{\ast}}\left(  E_{P^{\ast}}\left(  R^{T}\right)  R\right)  >0.
\end{align*}
It follows from A3) that $Var_{P^{\ast}}\left(  E_{P^{\ast}}\left(
R^{T}\right)  R\right)  >0$. Hence,%
\begin{align*}
L_{0}  &  =\sup_{\bar{\lambda}\in\mathbb{R}^{d}}\{\bar{\lambda}^{T}%
Z-\inf_{\bar{\Lambda}\in\mathbb{R}^{d\times d}}E_{P^{\ast}}[\left\Vert
\bar{\Lambda}R+\bar{\lambda}\right\Vert _{2}^{2}]\}\\
&  =\sup_{\bar{\lambda}\in\mathbb{R}^{d}}\{\bar{\lambda}^{T}Z-\left\Vert
\bar{\lambda}\right\Vert _{2}^{2}\left(  1-\left\Vert \mu_{\ast}\right\Vert
_{2}^{4}/\mu_{\ast}^{T}\Sigma_{\ast}\mu_{\ast}\right)  \}\\
&  =\frac{\left\Vert Z\right\Vert _{2}^{2}}{4\left(  1-\left\Vert \mu_{\ast
}\right\Vert _{2}^{4}/\mu_{\ast}^{T}\Sigma_{\ast}\mu_{\ast}\right)  }.
\end{align*}
It remains to identify $Z$. Observe that
\begin{align*}
\mu_{n}  &  =\rho1+2\left(  \Sigma_{n}\phi^{\ast}-\phi^{\ast T}\Sigma_{n}%
\phi^{\ast}1\right)  /\lambda_{1}^{\ast}\\
&  =\rho1+2\left(  \Sigma_{\ast}\phi^{\ast}-\phi^{\ast T}\Sigma_{\ast}%
\phi^{\ast}1\right)  /\lambda_{1}^{\ast}\\
&  +2\left(  H_{n}\phi^{\ast}-\phi^{\ast T}H_{n}\phi^{\ast}1\right)
/\lambda_{1}^{\ast}\\
&  =\mu_{\ast}+2\left(  H_{n}\phi^{\ast}-\phi^{\ast T}H_{n}\phi^{\ast
}1\right)  /\lambda_{1}^{\ast},
\end{align*}
where $H_{n}:=\Sigma_{n}-\Sigma_{\ast}$. By A1) we have
\begin{align*}
n^{-1/2}\sum_{i=1}^{n}\left(  R_{i}-\mu_{\ast}\right)   &  \Rightarrow
Z_{0}\sim N(0,\Upsilon_{g_{1}}),\\
n^{1/2}H_{n}  &  \Rightarrow Y\sim N(0,\Upsilon_{g_{2}}).
\end{align*}
Thus,
\begin{align*}
&  n^{-1/2}\sum_{i=1}^{n}\bar{\lambda}^{T}\left(  R_{i}-\mu_{\ast}\right)
+2n^{1/2}\bar{\lambda}^{T}\left(  H_{n}\phi^{\ast}-\phi^{\ast T}H_{n}%
\phi^{\ast}1\right)  /\lambda_{1}^{\ast}\\
&  \Rightarrow\bar{\lambda}^{T}Z=\bar{\lambda}^{T}\left(  Z_{0}+Z_{1}\right) ,
\end{align*}
where
\[
Z_{1}:=2\left(  Y\phi^{\ast}-\phi^{\ast T}Y\phi^{\ast}1\right)  /\lambda
_{1}^{\ast}.
\]

\end{document}